\documentclass{article}

\newcommand{\idspace}{\{0,1\}^d}

\newcommand{\dht}{{\scshape dht}}
\newcommand{\dhts}{{\dht s}}
\newcommand{\id}{{\scshape id}}
\newcommand{\ids}{{\id s}}
\newcommand{\xor}{{\scshape xor}}
\newcommand{\ptwop}{{P2P}}

\newcommand{\Thead}{{T_1^\prime}}
\newcommand{\Ttail}{{T_1^{\prime\prime}}}
\newcommand{\Ts}{{T_1^*}}
\newcommand{\Tss}{{T_1^{**}}}
\newcommand{\Th}{{{T}_1^*}}
\newcommand{\Thh}{{{T}_1^{**}}}
\newcommand{\Thhh}{{{T}_1^{***}}}

\newcommand{\ow}{{\overline{w}}}
\newcommand{\oDelta}{{\overline{\Delta}}}
\newcommand{\oV}{{\overline{V}}}

\usepackage[utf8]{inputenc}
\usepackage{amsmath}
\usepackage{amsthm}
\usepackage{array}
\usepackage{tikz}
\usepackage{float}
\usepackage[numbers]{natbib}
\usepackage[margin=1.45in]{geometry}
\usepackage[final]{changes}
\newcommand\cG{{\cal G}}

\newcommand\cL{{\cal L}}

\newcommand\cS{{\cal S}}

\newcommand{\E}[1]{{\textbf E}\left[#1\right]}
\newcommand{\e}{{\textbf E}}

\newcommand{\va}{{\textbf{Var}}}

\newcommand{\p}[1]{{\textbf P}\left\{#1\right\}}

\newcommand\inprobHIGH{\,{\buildrel p \over \rightarrow}\,} 
\newcommand\inprob{{\inprobHIGH}}

\newcommand{\eql}{\,{\buildrel \cL \over =}\,}

\newcommand{\eqd}{\,{\buildrel {\rm def} \over =}\,}

\newcommand{\ind}[1]{{\textbf{1}_{\left[#1\right]}}}

\newtheorem{lemma}{Lemma}
\newtheorem{theorem}{Theorem}

\title{The Analysis of Kademlia for random IDs}

\author{
    Xing Shi Cai \hspace{8 mm} Luc Devroye\thanks{Research of the authors was supported
    by NSERC.}\\
    \small School of Computer Science, McGill University of Montreal, Canada,\\
    \small \texttt{xingshi.cai@mail.mcgill.ca} \hspace{8 mm} \texttt{lucdevroye@gmail.com}
}

\begin{document}

\maketitle

\begin{abstract}
Kademlia~\citep{Maymounkov02} is the de facto standard searching algorithm for
\ptwop\ (peer-to-peer) networks on the Internet.  In our earlier
work~\citep{Cai2013}, we introduced two slightly different models for Kademlia
and studied how many steps it takes to search for a target node by using
Kademlia's searching algorithm. The first model, in which nodes of the network
are labeled with deterministic \ids, had been discussed in that paper.  The
second one, in which nodes are labeled with random \ids, which we call the
Random \id\ Model, was only briefly mentioned.  Refined results with detailed
proofs for this model are given in this paper. Our analysis shows that with
high probability it takes about $c \log n$ steps to locate any node, where $n$
is the total number of nodes in the network and $c$ is a constant that does not
depend on $n$.
\end{abstract}

\section{Introduction to Kademlia}

A \ptwop\ (peer-to-peer) network~\citep{Schollmeier2001} is a decentralized
computer network which allows participating computers (\emph{nodes}) to share
resources.  Some \ptwop\ networks have millions of live nodes.  To allow
searching for a particular node without introducing bottlenecks in the network,
a group of algorithms \added{called} \dht\ (Distributed Hash Table)~\citep{Balakrishnan03}
was invented in the early 2000s, including Plaxton's
algorithm~\citep{Plaxton1999accessing}, Pastry~\citep{Rowstron01}, {\scshape
can}~\citep{Ratnasamy2001}, Chord~\citep{Stoica2001},
Koorde~\citep{Kaashoek2003koorde}, Tapestry~\citep{Zhao04}, and
Kademlia~\citep{Maymounkov02}. Among them, Kademlia is most widely used in
today's Internet.

In Kademlia, each node is assigned an \id\ selected uniformly at random from
$\{0,1\}^d$ (\emph{\id\ space}), where $d$ is usually $128$~\citep{Steiner07} or
$160$~\citep{Crosby07}. The \emph{distance} between two nodes is calculated by
performing the bitwise exclusive or (\xor) operation over their \ids\ and taking the result
as a binary number. (In this work \emph{distance} and \emph{closeness} always
refer to the \xor\ distance between \ids.) 

Roughly speaking, a Kademlia node keeps a table of a few other nodes
(\emph{neighbors}) \replaced{whose distances are sufficiently diverse}{with diversified distances to
itself}.  So when a node
searches for an \id, it always has some neighbors close to its target.  By
inquiring these neighbors, and these neighbors' neighbors, and so on, the
node that is closest to the target \id\ in the network will be found
eventually. Other \dhts\ work in similar ways.  The differences mainly come
from how distance is defined and how neighbors are chosen. For a more detailed
survey of \dhts, see~\citep{Balakrishnan03}.

\section{The Random ID Model}

This section briefly reviews the Random \id\ Model for Kademlia defined
in~\citep{Cai2013}. Let $d \ge \log_2 n$ be the length of $n$ binary \ids\
$X_1,\ldots,X_n$ chosen uniformly at random from $\idspace$ without
replacement. Consider $n$ nodes indexed by $i \in \{1,\ldots,n\}$. Let $X_i$ be the
\id\ of node $i$. 

Given two \ids\ $x=(x_1,\ldots,x_d), y=(y_1,\ldots,y_d)$, their \xor\ 
distance is defined by 
$$
\delta(x,y) = \sum_{j=1}^d (x_j \oplus y_j) \times 2^{d-j}.
$$
where $\oplus$ is the \xor\ operator
\[
u \oplus v =
\begin{cases}
1  & \text{if $ u \ne v $,} \\
0 & \text{otherwise.}
\end{cases}
\]

Let $\ell(x,y)$ be the length of the common prefix of $x$ and $y$. The $n$
nodes can be partitioned into $d+1$ parts by their common prefix length with
$x$ via
$$
\cS(x,j) = \{i:1 \le i \le n,~\ell(x,X_i)=j\}, \qquad 0 \le j \le d.
$$
For each $1 \le i \le n$, $d$ tables (\emph{buckets}) of size \replaced{at most}{up to} $k$ are
kept, where $k$ is a fixed positive integer. Buckets are indexed by $j \in
\{0,\ldots,d-1\}$.  The bucket $j$ is filled with \replaced{$\min\{k, |\cS(X_i,j)|\}$}{up to \(k\)} indices drawn
uniformly at random from $\cS(X_i,j)$ without replacement. Note that the first
$j$ bits of $X_s$, if $s \in \cS(X_i,j)$, agree with the first $j$ bits of
$X_i$, but the $(j+1)$-th bit is different.

Searching for $y \in \idspace$ initiated at node $i$ proceeds as follows. Given
that $\ell(y,X_i) = j$, $y$ can only be in $\cS(X_i,j)$. Thus, all indices from
the bucket $j$ of $i$ are retrieved, say $i_1,\ldots,i_k$. From them, the one
having shortest distance to $y$ is selected as $i^*$. (In fact, any selection
algorithm would be sufficient for the results of this paper.) Note that
$$
\ell(y,X_{i^*}) = \max_{1 \le r \le k}\ell(y,X_{i_r}).
$$
\added{Thus the choice of \(i^*\) does not depend on the exact distances from \(X_{i_1},\ldots,
X_{i_k}\) to \(y\).}
Therefore, instead of the \xor\ distance, only the length of common prefix is
needed in the following analysis of searching.

The search halts if $y = X_{i}$ or if the bucket is empty. In the latter case,
$X_i$ is closest to $y$ among all nodes. Otherwise we continue from $i^*$.
Since $\ell(y,X_{i^*}) > \ell(y,X_i)$, the maximal number of steps before
halting is bounded by $d$.  Let $T_i$ be the number of steps before halting in
the search of $y$ when started from $i$ (\emph{searching time}).  Then $T_i =
T_{i^*}+1$.

Treating $X_1,\ldots,X_n$ as strings consisting of zeros and ones, they can be
represented by a tree data structure called
$\emph{trie}$~\citep{Szpankowski2011}.  The $\cS(x,j)$'s can be viewed as
subtrees.  Filling buckets is equivalent to choosing \replaced{at most}{up to} $k$ leaves from each
of these subtrees.  Fig.\,\ref{fig:partition} gives an example of an \id\ trie.

\begin{figure}
\centering {
    \scalebox{0.8} {
        \includegraphics{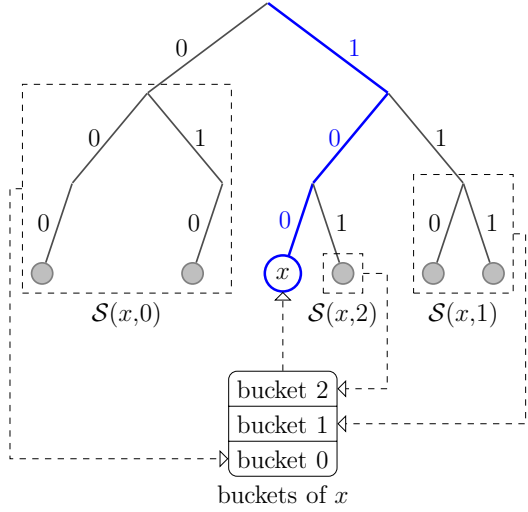}
    }
}
\caption[]{An example of Kademlia \id\ trie. Given an \id\ $x=(1,0,0)$, the
trie is partitioned into subtrees $\cS(x,0), \cS(x,1)$ and $\cS(x,2)$.  Node
$x$ maintains a bucket for each of these subtrees containing \replaced{at most}{up to} $k$ nodes
from the corresponding subtree.}
\label{fig:partition}
\end{figure}

\section{Main Results}

The structure of the model is such that nothing changes if $X_1,\dots,X_n, y$ are
replaced by their coordinate-wise \xor\ with a given vector $z \in \idspace$.
This is a mere rotation of the hypercube.
Thus, it can be assumed without loss of generality that $y =
(1,1,\ldots,1)$, the rightmost branch in the \id\ trie.

If $d \sim c \log_2 n$ for some $c \ge 1$, the searching time is \replaced{\(O(\log n)\)}{quite stable and
acceptable}, which is undoubtedly a contributing factor in Kademlia's success.
If $d = \omega(n)$, then it is not a useful upper bound of searching time
any more. However, in some
probabilistic sense, $T_i$ can be much smaller than
$\log_2 n$---it can be controlled by the parameter $k$, which measures the
amount of storage consumed by each node. The aim of this work is to investigate
finer properties of these random variables. In particular, the following
theorem is proved:
\begin{theorem}
\label{thm:random:target}
\added{Assume that \(d \ge \log_2 n\). Let \(k>0\) be a fixed integer.}
Let $\inprob$
denote convergence in probability. Then
\begin{align*}
& \frac {T_{1}} {\log_2 n} \inprob \frac 1 {\mu_k}, \qquad \text{as } n \to
\infty, \\
& \frac {\e T_{1}} {\log_2 n} \to \frac 1 {\mu_k}, \qquad \text{as } n \to
\infty,
\end{align*}
where $\mu_k$ is a function of $k$ only:
$$
\mu_k = \sum_{j=1}^{\infty} 1 - \left( 1 - \frac 1 {2^{j-1}} \right)^k.
$$
In particular, $\mu_1 = 2$.
\end{theorem}

\added{In the rest of the paper, we first show that once the search reaches a node that shares a
common prefix of length about \(\log n\) with \(y\), the search halts in \(o(\log n)\) steps.
Thus it suffices to prove Theorem \ref{thm:random:target} for the time that it takes for this
event to happen.  Then we show that the \id{} trie is well balanced with high probability. Thus
when \(n\) is a power of \(2\), we can couple the search in the original trie with a search in a
trie that is a complete binary tree. It proves the theorem for this special case. After that, we
give a sketch of how to deal with general \(n\). At the end we briefly summarize some
implications of the theorem.}

\section{The Tail of the Search Time}

To keep the notation simple, let $m = \log_2 n$ and \replaced{and note that \(m\) is not necessarily
integer-valued}{assume that $n$ is a power
of two. In Section~\ref{sec:general}, we sketch what to do when $n$ is not a
power of two}.  Also, for analytic
purposes, define
$$
J = \min\left\{j: \frac{n}{2^{j+1}} \le m^{4}\right\}.
$$
Since $n / 2^J > m^{4}$ and $n / 2^{J+1} \le m^{4}$, 
\begin{align}
    J & < \log_2 \frac{n}{m^4} = m - 4 \log_2 m \le m, 
    \label{eq:J:upper}\\
    J & \ge \log_2 \frac{n}{m^4}-1 = m - 4 \log_2 m - 1.
    \label{eq:J:lower}
\end{align}

The importance of $J$ follows from the fact that once the search reaches a node
$i$ with $\ell(X_i,y) \ge J$, it takes very few steps to finish. 
\added{Let $\Thead$ be}
the number of search steps that depart from a node in the set
$\cS(y,j)$ for some $j < J$, with the very first node in the search being $1$. 

\begin{lemma}
    Theorem~\ref{thm:random:target} follows if 
    $$
    \frac {\Thead}{\log_2 n} \inprob \frac 1 {\mu_k}, \qquad \text{as } n \to \infty.
    $$
    \label{lem:tail}
\end{lemma}
\begin{proof}
    \added{Let \(\Ttail = T_1 - \Thead\).}
    $\Ttail$ counts steps of the search departing from a node 
    in  $\bigcup_{j = J}^{d-1} \cS(y,j)$.  Thus
    $$
    \Ttail \le \sum_{j \ge J}^{d-1} \ind{|\cS(y,j)| > 0}.
    $$
    Noting that
    \begin{align}
        \e |\cS(y,j)| = \frac{n}{2^{j+1}},
        \label{eq:tree:expectation}
    \end{align}
    by linearity of expectation, 
    \begin{align*}
        \e \Ttail
        & \le \sum_{j \ge J}^{d-1} \p{|\cS(y,j)| \ge 1} 
        \le \sum_{j \ge J}^{d-1} \min\{\e |\cS(y,j)|, 1\}\\
        & \le \sum_{j \ge J}^{d-1} \min\left\{\frac{n}{2^{j+1}}, 1\right\} 
        \qquad (\text{by \eqref{eq:tree:expectation}})\\
        & \le \sum_{j \ge J}^{d-1} \ind{2^{j+1}<n} + \sum_{j \ge J}^{d-1}
        \ind{2^{j+1} \ge n} \times \frac{n}{2^{j+1}}\\
        & \le 4 \log_2 \log_2 n + 2 \qquad (\text{by \eqref{eq:J:lower}}).
    \end{align*}
    Thus, for all $\epsilon > 0$ fixed, 
    $$
    \p{\Ttail \ge \epsilon \log_2 n} 
    \le \frac{\e \Ttail}{\epsilon \log_2 n} 
    = o(1),
    $$
    Therefore $\Ttail/\log_2 n \inprob 0$. For the expectation, note that
    $$
    \frac {\e{\Thead}} {\log_2 n}  \to \frac{1}{\mu_k}, \qquad \text{as
    } n \to \infty,
    $$
    by the lemma's assumption and the fact that $\Thead / \log_2 n \le 1 <
    \infty$.
\end{proof}

\section{Good Tries and Bad Tries}

Since the tail of search does not matter, define a new partition $S_j$ of all nodes by
merging subtrees $\cS(y,j)$ for $j \ge J$ as follows:
$$
\cS_j =
\begin{cases}
    \cS(y,j) & \text{if $0 \le j < J$,} \\
    \bigcup_{i=J}^{d} \cS(y,i) & \text{if $j = J$}.
\end{cases}
$$
Let $N_j = |\cS_j|$. It follows from~\eqref{eq:tree:expectation} that
\begin{align}
    \e N_j =
    \begin{cases}
        {n}/{2^{j+1}}  & \text{if $0 \le j < J$,} \\
        {n}/{2^{J}} & \text{if $j = J$},
    \end{cases}
    \label{eq:N:expectation}
\end{align}
or simply $\e N_j = n/2^{(j+1)\wedge J}$, where $a \wedge b \eqd \min\{a,b\}$.
\added{Note that $N_j$ is
    hypergeometric with parameters
    $$
    \left(n,\frac{2^{d}}{2^{(j+1) \wedge J}}, 2^{d}-\frac{2^{d}}{2^{(j+1)
    \wedge J}}\right),
    $$
    i.e., it corresponds to the selection of $n$ balls without replacement from
    an urn of $2^{d}$ balls of which $2^{d}/2^{(j+1)\wedge J}$ are
    white~\citep[chap.~6.3]{Johnson2005}.}

The analysis of $\Thead$ can be simplified if the $N_j$'s are all close to their
expectations. To be precise, let $\alpha = m^{-3/2}$ be the \emph{accuracy
parameter}. An \id\ trie is \emph{good}, if
$$
\left| N_j - \e N_j \right| \le \alpha \times \e N_j,
$$
for all $0 \le j \le J$. Otherwise it is called \emph{bad}.

\begin{figure}
\centering {
    \scalebox{1.1} {
        \includegraphics{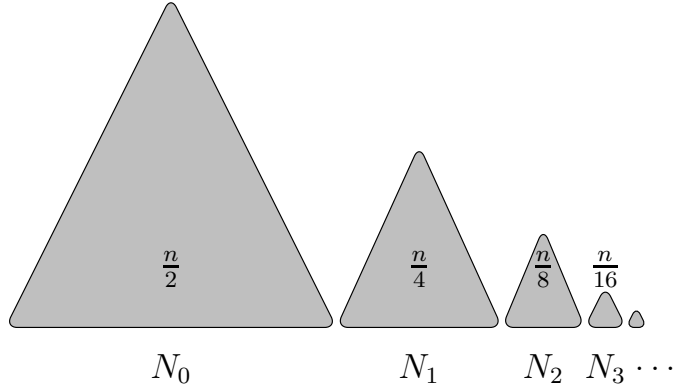}
    }
}
\caption[]{The approximate sizes of subtrees in a good trie.}
\label{fig:tree}
\end{figure}

\begin{lemma}
    \label{lem:pleasant}
    The probability that an \id\ trie is bad is $o(1)$.
\end{lemma}

\begin{proof}
    It follows from the union bound that
    \begin{align*}
        \p{\bigcup_{j=0}^J \left[ |N_j - \e N_j| > \alpha \times \e N_j \right]}
        & \le \sum_{j=0}^J \p{ \left[ |N_j - \e N_j| > \alpha \times \e N_j \right]} \\
        & \le \sum_{j=0}^J \frac{\va(N_j)}{(\alpha \times \e N_j )^2} \qquad \text{(by
        Chebyshev's inequality}) \\
        & \le \sum_{j=0}^J \frac{\e N_j}{(\alpha \times \e N_j )^2} \qquad
        \text{($N_j$ is hypergeometric)}\\
        & \le \frac{1}{\alpha^2} \times \sum_{j=0}^J \frac{2^{j+1}}{n} \qquad
        \text{(by \eqref{eq:N:expectation})}\\
        & \le m^3 \times \frac{2^{J+2}}{n}
        = o(1). \qquad \left(\text{since } \frac{n}{2^{J}} > m^4\right)
    \end{align*}
    The fact used here is that $\va(N_j) \le \va(N_j')$ where
    $N_j'$ is binomial $(n,1/2^{(j+1) \wedge J})$.
    For the binomial, $\va(N_j') \le \e
    N_j' = \e N_j$.
\end{proof}

\section{Proof when $n$ Is a Power of $2$}

In this section, $n$ is assumed to be a power of $2$, i.e., $m$ is an integer.
The general case is treated in the next section.

\subsection{A Perfect Trie}

\label{sec:perfect}

Construct a coupled \id\ trie consisting of $Y_1,\ldots,Y_n$ as follows. If
$N_j \ge \e N_j$, i.e., the size of the subtree $\cS_j$ is at least its
expectation, let $Y_i = X_i$ for the $\e N_j$ smallest indices in $\cS_j$.
After this preliminary coupling, some $Y_i$'s are undefined. The indices $i$
for which $Y_i$ are undefined go into a global pool $\cG$ of size
$$
\sum_{j=0}^J \max\{N_j - \e N_j, 0 \}.
$$
For a good trie, the size of the pool is at most
$$
\sum_{j=0}^J \alpha \times \e N_j = \alpha \times \e \sum_{j=0}^J N_j = \alpha n.
$$

For a subtree $\cS_j$ of size $N_j < \e N_j$, take $\E{N_j} - N_j$ indices $i$ from
$\cG$ and assign $Y_i$ a value, that is different from all other $Y_s$'s, and
that has $\ell(Y_i,y) \wedge J = j$.  Subtrees of this new trie have fixed
sizes of
\begin{align}
    |\{i:\ell(Y_i,y) \wedge J=j\}|
    = \e N_j = \frac{n}{2^{(j+1) \wedge J}},
    \qquad 0 \le j \le J.
    \label{eq:Y}
\end{align}
A trie like this is called \emph{perfect}.  Indices $i$ for which $X_i
\ne Y_i$, i.e., $i \in \cG$, are called \emph{ghosts}. Other indices are
called \emph{normal}.

Next, refill the buckets according to the perfect trie, but keep buckets of normal
indices containing no ghosts unchanged.  Observe that a search step departing
at a normal index $i$ proceeds precisely the same in both tries if bucket $j$
(with $j=\ell(Y_i,y)$) of $i$ does not contain ghosts.  \added{Assuming that the original trie is good,
the
probability that a bucket \added{that corresponds to \(\cS_j\) for some \(j \le J\)} contains a ghost is not more than $k
\alpha$. This is because in the newly constructed prefect trie, the subtree \(\cS_j\)
contains no more than \(\alpha\) proportion of ghost nodes.}

Let $\Ts$ denote the number of search steps starting from node $1$ via node $i$
with $\ell(Y_i,y) < J$ in the perfect trie. Then
$
\left[ \Ts \ne \Thead \right] \subseteq B,
$
where $B$ is the event that at least one node in the buckets
encountered during a search is a ghost. Let $A$ be the event that the
trie is good. It follows from Lemma~\ref{lem:pleasant} that
$$
    \p{\Ts \ne \Thead} \le \p{B}
    \le \p{B, A} + \p{A^c}
    \le J \times k \alpha + o(1)
    = o(1).
$$
Therefore, Theorem~\ref{thm:random:target} follows if
$$
\frac {\Ts} {\log_2 n} \inprob \frac 1 {\mu_k}, \qquad \text{as } n \to
\infty.
$$

\subsection{Filling the Buckets with Replacement}

\label{sec:replacement}

To deal with the problem that buckets are filled by sampling without
replacement, another coupling argument is needed. Let $p_j$ be the probability
that the $k$ items sampled with\deleted{out} replacement from a set of size $n/2^{j+1}$ are
not all distinctive. Observe that by the union bound,
$$
p_j \le \binom {k}{2} \frac{2^{j+1}}{n} \le \frac{k^22^j}{n}.
$$
If $\ell(Y_i,y) = j < J$, then bucket $j$ of $i$ has $k$ elements drawn without
replacement from
$$
\cS = \{s:\ell(Y_s,y) \ge j+1\}, \quad 0 \le j < J.
$$
Observe that
$$
|\added{\cS}| = \frac{n}{2^{j+2}} + \frac{n}{2^{j+3}} + \dots + \frac{n}{2^{J}} +
\frac{n}{2^{J}} = \frac{n}{2^{j+1}}.
$$
Hence, with probability $1-p_j$, the sampling can be seen as having been
carried out with replacement. 

The coupling is as follows: for all $i$ with $\ell(Y_i,y)=j$ and all $0 \le
j < J$, mark bucket $j$ of $i$ with probability $p_j$. When a bucket is
marked, replace its entries with $k$ new entires drawn with replacement
conditioned on the existence of at least one duplicate entry. In this way,
all bucket entries are for a sampling with replacement. Let the search time,
starting still from $1$, \added{be denoted }by $\Tss$. Let $D$ be the event that during the
search \replaced{a marked}{an unmarked} bucket is encountered. Observe that
$
\left[ \Ts \ne \Tss \right] \subseteq D.
$
Therefore
$$
\p{\Ts \ne \Tss} 
\le \p{D} 
\le \sum_{j=0}^{J-1} p_j 
\le \sum_{j=0}^{J-1} 
\frac{k^22^{j}}{n}
< \frac{k^22^{J}}{n}
< \frac{2 k^2}{m^{4}} = o(1).
$$
So Theorem~\ref{thm:random:target} follows if
$$
\frac {\Tss} {\log_2 n} \inprob \frac 1 {\mu_k}, \qquad \text{as } n \to
\infty.
$$

\subsection{Analyzing $\Tss$ Using a Sum of I.I.D.\ Random Variables}

\label{sec:geo1}

Let $\Delta_0 = \ell(Y_1,y)$. Assume that step $t$ of the search departs from node
$i$ and reaches node $i^{*}$. Let $\Delta_t = \ell(Y_{i^{*}},y)-\ell(Y_i,y)$,
i.e., $\Delta_t$ represents the progress in this step. Then
$$
    \Tss 
    = \inf\left\{t:\sum_{s=0}^t \Delta_{\added{s}} \ge J\right\}.
$$
Due to the recursive structure of a \replaced{perfect}{prefect} trie,
$\Delta_1,\Delta_2,\ldots$, although not i.i.d.,
should have very similar distributions. This intuition leads to the following
analysis of $\Tss$ by studying a sum of i.i.d.\ random variables.

One observation allows us to deal with truncated version of $\Delta_t$'s is as
follows:
\begin{lemma}
    Let $w_0,w_1,\ldots$ be a sequence of real numbers with \(\sum_{t \ge 0} w_t = \infty\). Define
    $$
    \ow_t = w_t \wedge \left(M - \sum_{s=0}^{t-1} \ow_s \right), \qquad
    t=0,1,2\ldots,
    $$
    where $M$ is also a real number.  Then
    $$
    \inf\left\{t:\sum_{s=0}^{t} w_{\added{s}} \ge M\right\}
    = \inf\left\{t:\sum_{s=0}^{t} \ow_{\added{s}} \ge M\right\},
    $$
    \added{where we define the infimum of an empty set to be \(\infty\).}
    \label{lem:truncation}
\end{lemma}

\begin{proof}
    Let $\tau = \inf\left\{t:\sum_{s=0}^{t} w_{\added{s}} \ge M\right\}$.
    \added{If \(\tau = \infty\) or \(\tau = 0\), the lemma is trivially true. 
    So we assume \(0 < \tau < \infty\).}
    By induction\added{ on \(t\)},
    one can show that $\ow_t=w_t$ if $t < \tau$. 
    \added{Since \(0 < \tau\), we have \(\ow_0 = w_0\), which is the induction basis.
        If \(\ow_s = w_s\) for all \(0 \le s \le t-1\) and \(t < \tau\), then
        $$
        \ow_t 
        = w_t \wedge \left(M - \sum_{s=0}^{t-1} \ow_s \right) 
        = w_t \wedge \left(M - \sum_{s=0}^{t-1} w_s \right) 
        = w_t.
        $$}Therefore $\sum_{s=0}^{t} w_{\added{s}}
    < \added{M}$ if and only if $\sum_{s=0}^{t} \ow_{\added{s}} < \added{M}$.
\end{proof}

\noindent Let $\oDelta_t = \Delta_t \wedge\left( J - \sum_{s=0}^{t-1} \oDelta_{\added{s}} \right)$.
It follows from the previous lemma that
\begin{align*}
    \Tss 
    = \inf\left\{t:\sum_{s=0}^t \oDelta_{\added{s}} \ge J\right\},
\end{align*}
which is quite convenient as the distribution of $\oDelta_t$ is easy to compute.

Assume again that step $t$ of the search departs from node $i$ with
$\ell(Y_i,y) = j< J$. Consider one item, say $z$, in bucket $j$ of $i$.
Recall that $z$ is selected uniformly at random from all \replaced{indices \(r\)}{indexes} with
$\ell(\added{r} ,y) \ge j+1$.  Thus it follows from the structure of a perfect
trie, which is given by~\eqref{eq:Y}, that
\begin{align*}
    & \p{\ell(Y_{z},y) = s} 
    = \frac{\frac{n}{2^{s+1}}}{\frac{n} {2^{j+2}}+ \frac{n} {2^{j+3}} + \dots +
    \frac{n}{2^{J}} + \frac{n}{2^{J}}}
    = \frac{1}{2^{s-j}}, \qquad j+1 \le s < J, \\
    & \p{\ell(Y_{z},y) \ge J} 
    = \frac{\frac{n}{2^{J}}}{\frac{n} {2^{j+2}}+ \frac{n} {2^{j+3}} + \dots +
    \frac{n}{2^{J}} + \frac{n}{2^{J}}}
    = \frac{1}{2^{J-j-1}}.
\end{align*}
Or shifted by $-j$,
\begin{align*}
    & \p{\ell(Y_{z},y) - j = s} 
    = \frac{1}{2^{s}}, \qquad 1 \le s < J-j, \\
    & \p{\ell(Y_{z},y) - j \ge J-j} 
    = \frac{1}{2^{J-j-1}}.
\end{align*}
If truncated by $J-j$, we obtain
\begin{align*}
    \p{(\ell(Y_{z},y) - j) \wedge (J-j) = s}
    \, = \frac{1}{2^{s \wedge (J-j-1)}},
    \qquad 1 \le s \le J-j.
\end{align*}
Note that this is \emph{exactly} the distribution of a geometric $(1/2)$ truncated by $J-j$.

Recall that among all the values of $\ell(\,\cdot,\,y)$ given by items in the
bucket $j$ of $i$, the one chosen as the next stop of the search gives the maximum.
Thus
$$
\Delta_t
= \max_{z \in \text{bucket } j}\{\ell(Y_{z},y) - j\}.
$$ 
Let $Z_1,Z_2,\ldots$ be i.i.d.\ geometric $(1/2)$.  Let $V =
\max\{Z_1,\ldots,Z_k\}$.
Then
\begin{align*}
    \oDelta_t
    & = \Delta_t \wedge (J-j) \\
    & = \max_{z \in \text{bucket } j}\{(\ell(Y_z, y) - j)\} \wedge (J-j) \\
    & = \max_{z \in \text{bucket } j}\{(\ell(Y_z, y) - j) \wedge (J-j)\} \\
    & \eql \max\{Z_1 \wedge (J-j),\ldots,Z_k \wedge (J-j)\} \\
    & = \max\{Z_1,\ldots,Z_k\} \wedge (J-j) \\
    & = V \wedge (J-j).
\end{align*}

Let $V_0$ be a geometric $(1/2)$ minus one. \added{Then $V_0 \wedge d \eql \Delta_0$.}
Let $V_1,V_2,\ldots$ be i.i.d.\ random variables distributed as $V$.  Let
$\oV_t = V_t \wedge (J- \sum_{s=0}^{t-1} \oV_s)$.  Using induction and the
previous argument about $\oDelta_t$, one can show that
\begin{equation}
\sum_{s=0}^{t} \oV_{\added{s}} \eql \sum_{s=0}^{t} \oDelta_{\added{s}} \qquad t=0,1,\ldots.
    \label{eq:V:D}
\end{equation}
\added{For the induction basis, note that
    \[
        \oDelta_0 = \Delta_0 \wedge J \eql (V_0 \wedge d) \wedge J = V_0 \wedge J = \oV_0.
    \]
    Assume that \(\sum_{s=0}^{t-1} \oV_s \eql \sum_{s=0}^{t-1} \oDelta_s\) for some \(t > 0\).
    Then for all \(0 \le i \le J\),}
    \begin{align*}
        \p{\sum_{s=0}^{t} \oDelta_s = i}
        &
        =
        \sum_{j=0}^{i} 
        \p{\oDelta_t = i-j~\left|~\sum_{s=0}^{t-1} \oDelta_s =j \right.}
        \p{\sum_{s=0}^{t-1} \oDelta_s =j}
        \\
        &
        =
        \sum_{j=0}^{i} 
        \p{V_t \wedge (J-j) = i-j}
        \p{\sum_{s=0}^{t-1} \oV_s =j}
        \\
        &
        =
        \sum_{j=0}^{i} 
        \p{\left[ \oV_t=i-\sum_{s=0}^{t-1} \oV_s \right] 
        \cap \left[ \sum_{s=0}^{t-1} \oV_s = j \right]}
        = \p{\sum_{s=0}^{t} \oV_s = i}
        .
    \end{align*}
    Thus \eqref{eq:V:D} is proved.
    It then follows from Lemma~\ref{lem:truncation} and \eqref{eq:V:D} that
\begin{align*}
    \Tss 
    \eql \inf\left\{t:\sum_{s=0}^t \oV_{\added{s}} \ge J\right\}
    = \inf\left\{t:\sum_{s=0}^t V_{\added{s}} \ge J\right\},
\end{align*}
which makes $\Tss$ much easier to analyze.

Since $V < s$ if and only if $Z_{1},\ldots,Z_{k}$ are all smaller than $s$, 
$$
\p{V < s} = \prod_{r=1}^{k} \p{Z_r < s} = \left(1-\frac{1}{2^{s-1}} \right)^{k}.
$$
Therefore, by definition of $\mu_k$,
$$
\e V = \sum_{s=1}^{\infty} \p{V \ge s} = \sum_{s=1}^{\infty} 1 -
\left(1-\frac{1}{2^{s-1}} \right)^{k} = \mu_k.
$$
Readers familiar with renewal theory~\citep[chap.\ 4.4]{Durrett2010probability}
can immediately see that
$$
\frac{\Tss}{\log_2 n} 
= \frac{\Tss}{J} \times \frac {J}{\log_2 n} \inprob \frac{1}{\e V} =
\frac{1}{\mu_k},
$$
which completes the proof of Theorem~\ref{thm:random:target} for $n$ which is
power of $2$. The following Lemma gives some more details.
\begin{lemma}
    If $\tau = \inf\left\{t:\sum_{s=0}^t V_{\added{s}}  \ge M\right\}$,
    $$
    \frac{\tau}{M/\e V} \inprob 1, \qquad \text{as
    } M \to \infty.
    $$
    \label{lem:convergence}
\end{lemma}
\begin{proof}
    Since $V_0+1$ is geometric $(1/2)$,
    $$
    \p{V_0 + 1 \le s} 
    = 1-\frac{1}{2^{s}} \ge \left( 1- \frac{1}{2^{s}} \right)^k = \p{V_1 \le s}.
    $$
    In other words, $V_0 \preceq V_1$, where $\preceq$ denotes
    stochastical ordering.  
    Let 
    \begin{align*}
        \tau' = \inf\left\{t:\sum_{s=1}^t V_s  \ge M\right\}, \qquad
        \tau'' = \inf\left\{t:\sum_{s=0}^t V_{s+1}  \ge M\right\} = \tau'-1.
    \end{align*}
    Then $\tau'' \preceq \tau$ and $\tau \le \tau'.$ By the strong law of large
    number\added{s}, both $\tau'/M$ and $\tau''/M$ converge\deleted{s} to $1/\e V$ almost surely.
    Therefore $\tau/M \inprob 1/\e V$.
\end{proof}

\section{Proof for the General Case}

\label{sec:general}

In this section, the proof Theorem~\ref{thm:random:target} for $n$ an
arbitrary integer is only sketched as most methods used here are very similar
to those in the previous section.

\subsection{An Almost Perfect Trie}

When $n$ is not power of $2$, $\e N_j = n/2^{(j+1) \wedge J}$ is not guaranteed
to be an integer. So a perfect trie is not well defined 
any more.  However, let us define
$$
b_j =
\begin{cases}
    \lceil \e N_j \rceil 
    = \lceil \frac{n}{2^{j+1}} \rceil  & \qquad 0 \le j < J, \\
    n - \sum_{s=0}^{J-1} b_s = n - \sum_{s=0}^{J-1} \lceil \frac{n}{2^{s+1}}
    \rceil & \qquad j = J.
\end{cases}
$$
Then the coupling argument for perfect tries used in Section~\ref{sec:perfect}
can still be applied, now replacing $\e N_j$ by $b_j$.

In this way, a trie consisting of $Y_1,\ldots,Y_n$ can be constructed, with
its subtrees having fixed sizes of
\begin{align}
    |\{i:\ell(Y_i,y) \wedge J=j\}|
    = b_j.
    \label{eq:Y:almost}
\end{align}
If the original trie is good, then the number of indices $i$ for which $X_i \ne
Y_i$, called \emph{ghosts}, is
bounded by
\begin{align*}
    \sum_{j=0}^{J-1} \alpha \times \e N_j  + (\alpha \e N_J + J) = \alpha n + J.
\end{align*}
A trie with these properties is called \emph{almost perfect}.

Let $\Th$ denote the number of search steps starting from node $1$ via node $i$
with $\ell(Y_i,y) < J$ in the almost perfect trie. If $\Th$ and $\Thead$ are
coupled  the same way as they were in Section~\ref{sec:perfect}, then
$
\left[ \Th \ne \Thead \right] \subseteq B,
$
where $B$ is the event that at least one node in the buckets
encountered during a search is a ghost. Let $A$ be the event that the
trie is good, which has probability $o(1)$ by Lemma~\ref{lem:pleasant}.
One can check that
\begin{align*}
    \p{\Th \ne \Thead}
    \le \p{B}
    \le \p{B, A} + \p{A^c}
    \le m k (m^{-3/2} + \frac{m}{2^{m}}) + o(1)
    = o(1).
\end{align*}
Again, Theorem~\ref{thm:random:target} follows if
$$
\frac {\Th} {\log_2 n} \inprob \frac 1 {\mu_k}, \qquad \text{as } n \to
\infty.
$$

\subsection{Filling the Buckets with Replacement}

The coupling argument used in Section~\ref{sec:replacement} to deal the problem
that buckets are filled by sampling without replacement can be adapted for an
almost perfect trie.  Let $p_j$ \deleted{is }be the probability that $k$ items sampled
without replacement from a set of size $b_{j+1} + \dots + b_{J}$ have
conflicts. Observe that, for $n$ large enough,
$$
b_{j+1} + \dots + b_{J} \ge \frac{n}{2^{j+1}} - (j+1) \ge
\frac{n}{2^{j+2}}.
$$
Thus it follows from the union bound that
$$
p_j 
\le \binom {k}{2} \frac{1}{b_{j+1}+\dots+b_{J}} 
\le \frac{k^{2}}{2 (b_{j+1} + \dots + b_{J})}
\le \frac{2^{j+1}}{n}.
$$

Let the search time of sampling without replacement be $\Thh$.  Let $\Thh$ and
$\Th$ be coupled in the same way as they were in Section~\ref{sec:replacement}.
Let $D$ be the event that during the search an unmarked bucket is encountered.
Since
$
\left[ \Th \ne \Thh \right] \subseteq D,
$
one can check that
\begin{align*}
    \p{\Th \ne \Thh}
    \le \p{D} \le \sum_{j=0}^{J-1} p_j
    < \frac{4 k^2}{m^{4}} = o(1).
\end{align*}
So once again, Theorem~\ref{thm:random:target} follows if
$$
\frac {\Thh} {\log_2 n} \inprob \frac 1 {\mu_k}, \qquad \text{as } n \to
\infty.
$$

\subsection{Analyzing $\Tss$ Using a Sum of I.I.D.\ Random Variables}

Consider two partitions of a line segment $L$ of length $n$.  From left to
right, cut $L$ into $J+1$ consecutive intervals $B_0,\ldots,B_J$, with
$|B_j|=b_j$, where $|a|$ denotes the length of $a$.  Again, from left to right,
cut $L$ into infinite many consecutive intervals $B'_0, B'_1, \ldots$, with
$|B_j'|=1/2^{j+1}$. 

Observe that for $0 \le j < J$, $B_j$ and $B'_j$ do not completely match since
$B_j$ is wider than $B_j'$.  However, since $|B_j| - |B_j'| \le 1$, for $0 \le
j < J$, the distance between the right endpoints of $B_j$ and $B_j'$ is at most
$J$. Therefore, the total length of unmatched regions, which are are
called \emph{death zones}, is $O(J^{2})$.

Let $\oDelta_0, \oDelta_1,\ldots$ and $V_0,V_1,\ldots$ be the same as in
Section~\ref{sec:geo1}.  A coupling between them can constructed as follows: pick
one point $z_0$ uniformly at random from the entire $L$. If $z_0$ falls in
interval $B_j$, let $\oDelta_0 = j$. If $z_0$ falls in interval $B'_j$, let
$V_0 = j$.  Note that $\oDelta_0 \eql \ell(Y_1,y)$.  Also note that since
$$
\p{V_0 = j} = \p{z_0 \in B'_j} = \frac{|B'_j|}{n} = \frac{1}{2^{j+1}}, \qquad
j=0,1,\ldots,
$$
$V_0$ is geometric $(1/2)$ minus one, as desired.

Assume that $\sum_{s=0}^{t-1} \oDelta_s = j$. Pick $k$ points from the line
segment starting from $B_{j+1}'$ to the right endpoint of $L$.  Let $V_t = s$
such that the rightmost one of the $k$ points falls into $B'_{j+s}$. Since
\begin{align*}
\p{V_t < s} 
& = \p{\text{all $k$ points are in $B'_{j+1}, \ldots
,B'_{j+s-1}$}} 
= \left( 1-\frac{1}{2^{s-1}} \right)^k,
\end{align*}
$V_t$ is again the maximum of $k$ i.i.d.\ geometric $(1/2)$.

If not all the $k$ points are in the range of $B_{j+1},\ldots,B_{J}$, keep
picking more points until $k$ of them are within this region.  Let
$\oDelta_t = s$ such that the rightmost of the these $k$ points falls into
$B_{j+s}$. Chosen in this way, $\oDelta_t$ has the same distribution as how
much progress one makes at step $t$ of the search.  Therefore 
$$ \Thh \eql
\inf\left\{t:\sum_{s=0}^t \oDelta_s \ge J\right\}. $$
It follows from Lemma~\ref{lem:convergence} that if
\begin{align*}
    \Thhh = \inf\left\{t:\sum_{s=0}^t V_s \ge J\right\},
\end{align*}
then ${\Thhh}/{\log_2 n} \inprob 1/{\mu_k}$ as $n \to \infty$.

Let $E$ be the event that at some step of the previous coupling, at least one
of the first $k$ chosen points falls into death zones. Note that $[\Thh \ne
    \Thhh] \subseteq E$. Therefore,
\begin{align*}
    \p{\Thh \ne \Thhh} 
    \le \p{E} 
    \le \sum_{j=0}^{J-1} k \frac{J^{2}}{b_{J}} \le
    \frac{m^{3}k}{m^{4}-m} = o(1).
\end{align*}
So the proof of Theorem~\ref{thm:random:target} \replaced{when \(n\) is}{for $n$ being} an arbitrary
integer is complete.

\section{Conclusions}

In a Kademlia system, \replaced{one often searches for a random \id{}}{which \id\ needs to be searched for is random}.  Although
$T_1$ is the searching time for a fixed \id, Theorem~\ref{thm:random:target}
still holds if the target $y$ is chosen uniformly at random from $\idspace$.

If $d \sim c \log_2 n$ with $c > 2$, there is no essential difference between
sampling the $n$ \ids\ with or without replacement from $\idspace$ as the
probability of a collision in sampling with replacement is $o(1)$. This is the
well known \emph{birthday problem}. Since in practice, a Kademlia system hands
out a new \id\ without checking its uniqueness, it is \deleted{is }wise to have $c > 2$
since then a randomly generated \id\ \replaced{clashes}{does not clash} with any existing \id\
\deleted{except }with \deleted{a }very small probability.

\added{Recall that $ \mu_k = \sum_{j=1}^{\infty} 1 - \left( 1 - 1/{2^{j-1}} \right)^k$. }
\added{Since the terms in the sum decrease in \(j\), \(\mu_k\)} can be bounded:
\begin{align*}
    \mu_k & \ge \int_{0}^\infty 1 - \left(1-\frac 1 {2^{x}}\right)^{k}
    \mathrm{d}x = \frac {H_k} {\log 2}, \\
    \mu_k & \le \int_{0}^\infty 1 - \left(1-\frac 1 {2^{x}}\right)^{k} \mathrm{d}x + 1 =
    \frac {H_k} {\log 2} + 1.
\end{align*}
Here $\log w$ denotes the natural logarithm of $w$, and $H_k = \sum_{s=1}^k
1/s$ is the $k$-th harmonic number. Since $H_k \sim \log k$,
$$
\lim_{k \to \infty} \frac {\mu_k} {\log_2 k} 
= \lim_{k \to \infty} \frac {H_k} {\log 2 \times \log_2 k} 
= 1.
$$
Thus, $T_1 /\log_k n \inprob \log_2 k / \mu_k = 1 + o_k(1)$. \added{Since \(T_1/(\frac 1 2 \log_2 n) \inprob
1\) when \(k=1\)}, an increase in
storage by a factor of $k$ \deleted{thus }results in a modest decrease in searching time by
a factor of $\log(k)/ \added{(2 \log 2)}$.

In~\citep{Cai2013}, it has been proved that if $X_1=x_1,\dots,X_n=x_n$ for fixed
$x_1,\ldots,x_n$, then
$$
\sup_{x_1,\ldots,x_n} \sup_{i} \sup_{y} \e{T_i} \le \left(\frac {\log 2}{H_k} +
o(1)\right)\log_2 n.
$$
Thus Theorem~\ref{thm:random:target} implies that the above upper bound is not
far from tight when $k$ is large. Table~\ref{table:mu:k} lists the numeric
values of $1/\mu_k$ and $\log(2)/H_k$ for $k = 1,\ldots,10$.

\begin{table}
    \label{table:mu:k}
    \setlength{\tabcolsep}{12pt}
    \caption{Numeric values of $1/\mu_k$ and $\log(2)/H_k$.}
    \centering {
        \begin{tabular}{ l l l }
        \hline\noalign{\smallskip}
        \multicolumn{1}{c}{$k$} & \multicolumn{1}{c}{$1/\mu_k$} &
        \multicolumn{1}{c}{$\log(2)/H_k$}         \\
        \noalign{\smallskip}
        \hline
        \noalign{\smallskip}
        1  & 0.5000000000 & 0.6931471806      \\
        2  & 0.3750000000 & 0.4620981204      \\
        3  & 0.3181818182 & 0.3780802804      \\
        4  & 0.2853260870 & 0.3327106467      \\
        5  & 0.2635627530 & 0.3035681083      \\
        6  & 0.2478426396 & 0.2829172166      \\
        7  & 0.2358018447 & 0.2673294911      \\
        8  & 0.2261891923 & 0.2550344423      \\
        9  & 0.2182781689 & 0.2450176596      \\
        10 & 0.2116151616 & 0.2366523364      \\
        \hline
        \end{tabular}
    }
\end{table}

If $k = \Theta(\log n)$, then $T_1 \sim \log n / {\log \log n}$ in probability as
$n \to \infty$.  The proof of Theorem~\ref{thm:random:target} is for fixed $k$
only, but one can verify that\added{ only minor changes are needed to make it work for} such modest
increase in $k$ as a function of $n$\deleted{ does not alter the results}.\added{ More specifically, to
make the coupling with searching in a perfect trie work, we only need to redefine \(J=\min\{j:1/2^{j+1} < m^7\}\) and
\(\alpha = m^{-3}\). And Lemma \ref{lem:convergence} needs to use a version of the weak law of
large numbers \cite[thm. 2.2.4]{Durrett2010probability} instead of the strong law of large numbers to
deal with the fact that \(\e V\) is not a constant anymore.}

\deleted{A more direct exercise shows that for $k = n^{\Theta(1)}$, $T_1 = \Theta(1)$ in
probability. }\added{If \(k = n^{\Theta(1)}\), we can show that
\(T_1 = \Theta(1)\) in probability.  Note that here only an upper bound of \(T_1\) is needed.
Assuming that the \id{} trie is good, it can be proved that in each search step the length of the common
prefix of the current node and the target node increases by at least \(c \log n\) with high
probability, where \(c\) is a constant depending on \(k\). Thus after at most \(O(1)\) steps, the
current node and the target node are both in a subtree of size at most \(k\). Then the search
terminates after one more step.}

\bibliography{citation}{}
\bibliographystyle{abbrplainnat}

\end{document}